\documentclass[letterpaper, 10 pt, conference]{ieeeconf}  

\IEEEoverridecommandlockouts                              

\usepackage{amsthm}
\usepackage{amsmath,amsfonts}
\usepackage{algorithmic}
\usepackage{algorithm}
\usepackage{array}
\usepackage{textcomp}
\usepackage{stfloats}
\usepackage{url}
\usepackage{verbatim}
\usepackage{graphicx}
\usepackage{cite}

\usepackage{hyperref}
\usepackage{color}

\usepackage{siunitx}
\usepackage{bbm}
\usepackage{numprint}
\npthousandsep{,}\npthousandthpartsep{}\npdecimalsign{.}

\newtheorem{proposition}{Proposition}
\newtheorem{remark}{Remark}

\newcommand{\ldn}[1]{{\color{black}#1}}
\newcommand{\zak}[1]{{\color{cyan}#1}}

\newcommand{\new}[1]{\textcolor{black}{#1}}\newcommand{\old}[1]{\iffalse {#1} \fi}

\title{\LARGE \bf Stable Linear Subspace Identification: A Machine Learning Approach}

\author{Loris Di Natale,$^\dagger$ Muhammad Zakwan,$^\dagger$ Bratislav Svetozarevic, \\Philipp Heer, Giancarlo Ferrari-Trecate, Colin N. Jones
\thanks{This research was supported by the Swiss National Science Foundation under NCCR Automation, grant agreement 51NF40\_180545.}
\thanks{L. Di Natale, B. Svetozarevic, and P. Heer are with the Urban Energy Systems Laboratory, Swiss Federal Laboratories for Materials Science and Technology (Empa), 8600 D\"{u}bendorf, Switzerland. L. Di Natale, M. Zakwan, G. Ferrari-Trecate, and C.N. Jones are with the Laboratoire d'Automatique, Swiss Federal Institute of Technology Lausanne (EPFL), 1015 Lausanne, Switzerland.}
\thanks{$^\dagger$ L. Di Natale and M. Zakwan contributed equally to this work. Corresponding author: L. Di Natale, {loris.dinatale@alumni.epfl.ch}.}}

\begin{document}

\maketitle
\thispagestyle{empty}
\pagestyle{empty}

\begin{abstract}

Machine Learning (ML) and linear System Identification (SI) have been historically developed independently. In this paper, we leverage well-established ML tools --- especially the automatic differentiation framework --- to introduce SIMBa, a family of discrete linear multi-step-ahead state-space SI methods using backpropagation. SIMBa relies on a novel Linear-Matrix-Inequality-based free parametrization of Schur matrices to ensure the stability of the identified model. 

We show how SIMBa generally outperforms traditional linear state-space SI methods, and sometimes significantly, although at the price of a higher computational burden. This performance gap is particularly remarkable compared to other SI methods with stability guarantees, where the gain is frequently above $25\%$ in our investigations, hinting at SIMBa's ability to simultaneously achieve state-of-the-art fitting performance and enforce stability. Interestingly, these observations hold for a wide variety of input-output systems and on both simulated and real-world data, showcasing the flexibility of the proposed approach. We postulate that this new SI paradigm presents a great extension potential to identify structured nonlinear models from data, and we hence open-source SIMBa on \href{https://github.com/Cemempamoi/simba}{https://github.com/Cemempamoi/simba}.


\end{abstract}

\section{Introduction}

While linear System Identification (SI) matured decades ago~\cite{ljung1998system}, Machine Learning (ML) only rose to prominence in recent years, especially following the explosion of data collection and thanks to unprecedented computational power. In particular, large Neural Networks (NNs) have shown impressive performance on a wide variety of tasks~\cite{goldberg2022neural,degrave2022magnetic}, leading to the recent boom of Deep Learning (DL) applications~\cite{lecun2015deep}. A key factor behind these successes has been the availability of efficient open-source libraries greatly accelerating the deployment of NNs, such as \texttt{PyTorch} and \texttt{TensorFlow} in Python. In particular, Automatic Differentiation (AD), at the core of the backpropagation algorithm~\cite{werbos1974beyond}, the backbone of NN training, nowadays benefits from extremely efficient implementations.

Given their effectiveness at grasping complex nonlinear patterns from data, NNs have recently been used for nonlinear system identification, where traditional SI methods struggle to compete~\cite{ljung2020deep, forgione2021continuous, brunton2022data}. NNs can be leveraged to create deep state-space models~\cite{gedon2021deep}, deep subspace encoders~\cite{beintema2023deep}, or deep autoencoders~\cite{masti2021learning}, for example. 
While applying NNs to identify nonlinear systems can achieve good performance, it can underperform for linear systems, where methods assuming model linearity might achieve better accuracy~\cite{gedon2021deep}. 

Although nonlinear SI has attracted a lot of attention in the last years, the identification of Linear Time Invariant (LTI) models is, however, still of paramount importance to many applications. Indeed, linear models come with extensive theoretical properties~\cite{rugh1996linear} and lead to convex optimization problems when combined with convex cost functions in a Model Predictive Controller (MPC)~\cite{muske1993model}, for example. Moreover, to date, numerous industrial applications still rely on the availability of linear models to conduct simulations, perform perturbation analysis, or design robust controllers following classical model-based techniques, such as $\mathcal{H}_2$, $\mathcal{H}_\infty$, and $\mu$-synthesis~\cite{khalil1996robust}.

In this work, we show how one can leverage ML tools --- backpropagation and unconstrained Gradient Descent (GD) --- for the identification of stable linear models, presenting a novel toolbox of \textbf{S}ystem \textbf{I}dentification \textbf{M}ethods leveraging \textbf{Ba}ckpropagation (SIMBa). Our research is related to the efforts in~\cite{schoukens2021improved, schulze2022data,drgovna2021deep}, where backpropagation was also used to identify LTI state-space models but only as part of specific frameworks and without considering stability constraints.

    \subsection{Subspace identification for linear systems}

State-of-the-art implementations of linear state-space SI often rely on subspace identification, such as the acclaimed MATLAB system identification toolbox~\cite{ljung1995system} or the \texttt{SIPPY} Python package~\cite{armenise2018open}. Both of them provide the three \textit{traditional} Subspace Identification Methods (SIMs), namely N4SID, MOESP, and CVA~\cite{qin2006overview}. 

Next to these traditional methods, \texttt{SIPPY} also proposes an implementation of \textit{parsimonious} SIMs (PARSIMs), namely  PARSIM-S~\cite{qin2005novel}, \text{PARSIM-P}~\cite{qin2003parallel}, and PARSIM-K~\cite{pannocchia2010predictor}, which enforce causal models by removing non-causal terms. While the former two methods do not work with closed-loop data since they assume no correlation between the output noise and the input, PARSIM-K was specifically designed to alleviate this assumption.

    \subsection{Enforcing stability}

In practice, when the true system is known to be stable, one usually requires the identified model to also be stable~\cite{manchester2021contraction}. 
In the case of traditional methods, a stability test can be done \textit{a posteriori} 
and a stable 
model can always be recovered from the extended observability matrix~\cite{maciejowski1995guaranteed}, for example, but at the cost of a sometimes significant performance loss~\cite{armenise2018open}. Note that such a correction cannot be implemented for PARSIM methods. 

Apart from such \textit{post-hoc} modifications to ensure stability, one can also modify the Least Squares (LS) estimation at the heart of many identification procedures, either introducing custom weighting factors~\cite{huang2016learning} or rewriting it as a constrained optimization problem~\cite{boots2007constraint, lacy2003subspace}.

Alternatively, one can leverage parametrizations of stable matrices, such as the ones proposed in~\cite{gillis2020note, jongeneel2022efficient}, and then use projected gradients to approximate the LS solution while ensuring the resulting model remains stable at each step, for example~\cite{mamakoukas2023learning}. A similar idea was utilized in~\cite{drgovna2021physics}, where the Perron-Frobenius theorem was leveraged to bound the eigenvalues of $A$ and hence ensure the system remains stable \textit{at all times}, even during the learning phase.

Finally, instead of directly constraining the state-space matrices, one can also simultaneously learn a model and a corresponding Lyapunov function for it, typically NN-based, thereby ensuring its stability by design~\cite{kolter2019learning}. This approach presents the advantage of naturally extending to nonlinear SI, contrary to all the others, but comes with a significant computational burden. Note that while SIMBa does not explicitly learn a Lyapunov function, it implicitly defines one to guarantee stability (see Section \ref{sec:theory}). However, instead of learning it with an NN, we leverage Linear Matrix Inequalities (LMIs) to parametrize Schur matrices, inspired from~\cite{revay2023recurrent}.

    \subsection{Contribution}

This paper introduces SIMBa, a linear SI toolbox that guarantees the stability of the identified model through a novel free parametrization of Schur matrices. Leveraging ML tools, it is able to minimize the multi-step-ahead prediction error to improve upon the performance of traditional SI methods. Our main contributions can be summarized as follows:
\begin{itemize}
    \item We propose SIMBa, a framework leveraging backpropagation and unconstrained GD for stable linear SI.
    \item We present a ready-to-use open-source Python implementation of SIMBa with a MATLAB interface. SIMBa is system-agnostic: it can seamlessly identify multi-input-multi-output systems, optimize different multi-steps-ahead performance metrics, deal with missing data and multiple trajectories, and it comes with smooth Graphical Processing Unit (GPU)-integration.
    \item We conduct an extensive empirical investigation of SIMBa's potential, exemplifying its flexibility and ability to attain state-of-the-art performance in various contexts, both on simulated and real-world data. \new{Overall, we often observe an improvement of over $25\%$ compared to the benchmark methods.}
\end{itemize}
Altogether, we propose a new paradigm for the identification of stable linear models. While it comes with an additional computational burden, SIMBa very often outperforms traditional methods in our experiments. Furthermore, it presents interesting extension potential, for example, to integrate prior knowledge about the system, include tailored nonlinearities, similarly to what was done in~\cite{zakwan2022physically}, or facilitate stable Koopman-based approaches like~\cite{loyakoopman, schulze2022data}.


\textit{Organization:} After a few preliminaries in Section~\ref{sec:preliminaries}, Section~\ref{sec:theory} presents the novel free parametrization of stable matrices used in SIMBa. Section~\ref{sec:simba} then introduces the SIMBa toolbox and Section~\ref{sec:results} provides empirical performance analyses. Finally, Section \ref{sec:conclusion} concludes the paper.

  \section{Preliminaries}
\label{sec:preliminaries}
Throughout this work, we are interested in the identification of discrete-time linear systems of the form
\begin{subequations}\label{eq:sys}
\begin{align}
    x_{k+1} &= Ax_k + Bu_k + w_k\label{eq:sys1}\\
    y_k &= Cx_k + Du_k + v_k \;, \label{eq:sys2}
\end{align}
\end{subequations}
where $x\in\mathbb{R}^n$ represents the state of the system, $u\in\mathbb{R}^m$ its input, $y\in\mathbb{R}^p$ its output, and $w_k \in \mathbb{R}^n$ and $v_k \in \mathbb{R}^p$ process and measurement noise, respectively. In general, state-space SI methods identify the unknown matrices $A\in\mathbb{R}^{n \times n}$, $B\in\mathbb{R}^{n \times m}$, $C\in\mathbb{R}^{p \times n}$, and $D\in\mathbb{R}^{p \times m}$ from a data set $\mathcal{D} = \{\left(u(0), y(0)\right), ..., \left(u(l_s), y(l_s)\right)\}_{s=1}^N$ of $N$ input-output measurement trajectories $s$ of length $l_s$. 

In our experiments, we split the data into a training, a validation, and a test set of trajectories respectively denoted $\mathcal{D}_{\textit{train}}$, $\mathcal{D}_{\textit{val}}$, and $\mathcal{D}_{\textit{test}}$. The first set is used to fit the model, the second one to assess when to stop the algorithm so as to avoid overfitting the training data, and the third one to measure the performance of the identified model on unseen data, as classically done in ML pipelines~\cite{lones2021avoid}.

Since the stability of identified models is often crucial, we are particularly interested in stable-by-design identification procedures. For systems of the form~\eqref{eq:sys}, stability is equivalent to the matrix $A$ being Schur, i.e., ensuring that all its eigenvalues $\lambda_i(A)$ be of magnitude smaller than one:
\begin{equation*}
    |\lambda_i(A)| < 1, \forall i=1,...,n \;.
\end{equation*}

\textit{Notation:} Given a symmetric matrix $F\in\mathbb{R}^{q\times q}$, $F\succ 0$ means $F$ is positive definite. For a matrix $V\in\mathbb{R}^{2q \times 2q}$, we define its block components $V_{11}, V_{12}, V_{21}, V_{22}\in\mathbb{R}^{q \times q}$ as:
\begin{equation}
    \begin{bmatrix}
        V_{11} & V_{12} \\ V_{21} & V_{22}
    \end{bmatrix}
    := V \;.
\end{equation}


    \section{Free Parametrization of Schur Matrices}
    \label{sec:theory}

To run unconstrained GD in the search space and hence take full advantage of \texttt{PyTorch} without jeopardizing stability, we need to construct $A$ matrices that are Schur \textit{by design}.
The following proposition provides such an LMI-based free parametrization of Schur matrices, inspired by~\cite{revay2023recurrent}.
\begin{proposition}\label{prop:generic}
    For any $W\in\mathbb{R}^{2n \times 2n}$, $V\in\mathbb{R}^{n \times n}$, $\tilde{\epsilon}\in\mathbb{R}$, and given $0< \gamma \leq 1$, let 
    \begin{equation}
    S:=W^\top W + \epsilon\mathbb{I}_{2n} \;, \label{eq:S generic}
    \end{equation}
    for $\epsilon = \exp{(\tilde{\epsilon})}$. Then
    \begin{equation} 
        A = S_{12} \left[\frac{1}{2}\left(\frac{S_{11}}{\gamma^2} + S_{22}\right) + V - V^\top\right]^{-1} \label{eq:A generic}
    \end{equation}
    is Schur stable with 
    $|\lambda_i(A)|<\gamma, \forall i=1,...,n$.
\end{proposition}
\begin{proof}
   We want to ensure that the magnitude of all the eigenvalues of $A$ is bounded by $\gamma$. From \cite[Theorem 2.2]{chilali1996h}, we know this is the case if and only if the following LMI holds for some symmetric $Q=Q^\top \succ 0$:
\begin{equation*}
    \begin{bmatrix}
    \gamma  Q & AQ \\ QA^\top &  \gamma Q
    \end{bmatrix} \succ 0 \; .
\end{equation*}
Taking its Schur complement, this is equivalent to 
\begin{align}
    \gamma  Q - A\frac{Q}{\gamma} A^\top &\succ 0 \label{eq:gamma} \;.
\end{align}
Defining the transformation $T = [I, -A]$ and introducing a free parameter $G\in\mathbb{R}^{n \times n}$, this can be rewritten as
\begin{align}
    &\gamma  Q -AGA^\top -A^\top G^\top A \nonumber\\
    &\qquad +AGA^\top +A^\top G^\top A  - A\frac{Q}{\gamma} A^\top \succ 0 \nonumber\\
    \iff &T \begin{bmatrix}
        \gamma  Q & AG \\ G^\top A^\top & G^\top + G - \frac{Q}{\gamma} 
    \end{bmatrix} T^\top \succ 0 \nonumber\\
    \iff & \begin{bmatrix}
        \gamma  Q & AG \\ G^\top A^\top & G^\top + G - \frac{Q}{\gamma} 
    \end{bmatrix} \succ 0 \;. \label{eq:generic LMI}
\end{align}
In words, $A$ is Schur stable with eigenvalues bounded by $\gamma$ if and only if there exist $Q \succ 0$ and $G$ such that~\eqref{eq:generic LMI} holds.

Let us now parametrize the left-hand side of the above LMI by the matrix $S$ in~\eqref{eq:S generic}. Remarkably, since $S$ is positive definite by construction for any choice of $W$, \eqref{eq:generic LMI} will always be satisfied, ensuring the stability of $A$. 

Finally, define
\begin{align*}
    S_{11} &:= \gamma Q, & S_{12} &:= AG\;, \\ 
    S_{21} &:= G^\top A^\top, & S_{22} &:= G^\top + G - \frac{Q}{\gamma} \;.
\end{align*}
This allows us to recover $Q = \frac{S_{11}}{\gamma}$, which is positive definite and symmetric by construction. We then note that
\begin{equation}
    G^\top + G = \frac{Q}{\gamma} + S_{22} = \frac{S_{11}}{\gamma^2} + S_{22} \label{eq:G}
\end{equation} 
needs to hold. Since $S_{11}$ and $S_{22}$ are symmetric, \eqref{eq:G} holds for any $V \in \mathbb{R}^{n \times n}$ if we set
\begin{equation*}
    G = \frac{1}{2}\left(\frac{S_{11}}{\gamma^2} + S_{22}\right) + V - V^\top \;.
\end{equation*}
Remembering that $A = S_{12}G^{-1}$ concludes the proof.
\end{proof}

Proposition~\ref{prop:generic} implies that we can run unconstrained gradient descent on two matrices $W$ and $V$ and \textit{always} construct a Schur matrix $A$ from them as in~\eqref{eq:A generic}. Furthermore, $\gamma$ allows us to tune the magnitude of the largest eigenvalue of $A$. 


\begin{remark}
    \ldn{Let $A$ be a Schur matrix. Setting $\gamma=1$ in the above proof yields that $A$ satisfies~\eqref{eq:generic LMI} for some $Q \succ 0$ and any $G\in\mathbb{R}^{n\times n}$. By definition, there exists $\epsilon>0$ such that
    \begin{equation*}
        \begin{bmatrix}
        \gamma  Q & AG \\ G^\top A^\top & G^\top + G - \frac{Q^\top}{\gamma} 
    \end{bmatrix} - \epsilon \mathbb{I}_{2n} =: \Gamma \succ 0 \;.
    \end{equation*}
    Define $V=0$, $W=\Gamma^{\frac{1}{2}}$, and $\tilde{\epsilon}=\log{\epsilon}$. Then $A$ can be constructed as in~\eqref{eq:A generic}, with $S$ as in~\eqref{eq:S generic}, showing that the proposed parametrization captures \emph{all} Schur matrices.}
\end{remark}

 \section{A Short Introduction to SIMBa}
    \label{sec:simba}

This section briefly describes SIMBa, discussing how to use some of its most critical \texttt{parameters}. The core implementations rely on the Python \texttt{PyTorch} library, but we also provide a MATLAB interface, and the whole package can be found on \href{https://github.com/Cemempamoi/simba}{https://github.com/Cemempamoi/simba}. The default value of each \texttt{parameter} has been empirically tuned to achieve robust performance, and more details on how to use SIMBa can be found in~\cite{tech}.

    \subsection{Optimization problem}
    \label{sec:opt prblm}
    
SIMBa minimizes the multi-step-ahead prediction error: 
\begin{align}
    \min_{A,B,C,D,x_{0}^{(s)}} &\quad \frac{1}{|Z|}\sum_{s\in Z}\left[\frac{1}{l_s}\sum_{k=0}^{l_s} m_k^{(s)}\mathcal{L}_{\textit{train}} \left( y^{(s)}(k), y^{(s)}_k \right)\right] \label{eq:obj io}\\
    \text{s.t.} &\quad  y^{(s)}_k = Cx^{(s)}_k + Du^{(s)}(k) \label{eq:y}\\
                &\quad  x^{(s)}_{k+1} = Ax^{(s)}_k + Bu^{(s)}(k) \;, \label{eq:x}
\end{align}
where 
$Z\in\mathcal{D}_{\textit{train}}$ is a randomly sampled batch of trajectories, naturally handling the case where multiple training trajectories are provided, 
and $\mathcal{L}_\textit{train}$ is the training loss. Note that, if $x_0$ is known, \texttt{learn\_x0} can be disabled. Finally, $m_k^{(s)}$ is a binary variable defaulting to one, \new{where the superscript $(s)$ stands for the sampled trajectory and the subscript $k$ for the time step,} that can be used in two ways. First, $m_k^{(s)}=0$ if $y^{(s)}(k)$ does not exist in the data, giving SIMBa a natural means to deal with missing values. Second $m_k^{(s)}=0$ with probability $0\leq p<1$, where $p$ is controlled through the \texttt{dropout} parameter, which is used to regularize the training procedure. This avoids overfitting the training data by randomly dropping points along the trajectories and might also provide empirical robustness to outliers present in the data.

Throughout our experiments, we use the Mean Square Error (MSE) loss, i.e., $\mathcal{L}_\textit{train}(y,\hat{y}) = (y-\hat{y})^2$, but it could seamlessly be substituted by any other metric of interest though \texttt{train\_loss}. For instance, one could use the Mean Absolute Error (MAE), which is more robust against outliers. 

SIMBa iteratively runs GD on batches of training data in~\eqref{eq:obj io} until \texttt{max\_epochs} epochs have passed, where an \textit{epoch} represents the fact that every training trajectory has been seen once. After each epoch, we evaluate the performance of the current model on the validation set. \new{To that end,} in the experiments in Section \ref{sec:results}, we set \new{$Z=\mathcal{D}_{\textit{val}}$ and }$\mathcal{L}_\textit{val}=\mathcal{L}_\textit{train}$ in~\eqref{eq:obj io}, but custom losses can be passed through \texttt{val\_loss}. At the end of the training, we keep the model performing best on the validation data and test it on $\mathcal{D}_{\textit{test}}$. 
Note that due to the nonconvexity of~\eqref{eq:obj io}, we cannot expect SIMBa to converge to a global minimum but only a local one, rendering it typically sensitive to initializations and some hyperparameters. 


\old{\begin{remark}
    Traditional identification methods can be seen as a special case of SIMBa: they rely on the one-step-ahead MSE, i.e., they treat each sample as a separate trajectory, setting $\mathcal{Z} = \mathcal{D}_\textit{train}$, $l_s=1$, and $\mathcal{L}_\textit{train}(y,\hat{y}) = (y-\hat{y})^2$. However, the stability of $A$ cannot be ensured during learning for traditional SIMs, contrary to SIMBa, which does not require any \emph{post-hoc} correction when Proposition~\ref{prop:generic} is leveraged.
\end{remark}}

    \subsection{Initialization}
	
To accelerate the convergence of SIMBa, it is possible to initialize the state-space matrices from the best solution found either by MATLAB or \texttt{SIPPY} --- hereafter referred to as the \textit{initialization method} --- with \texttt{init\_from\_matlab\_or\_ls=True}.\footnote{Since several benchmark methods are available, SIMBa uses the one achieving the best performance on the validation data.} 

Since $A$ is usually required to be stable, relying on the construction introduced in Proposition~\ref{prop:generic}, it cannot directly be initialized with the desired matrix $A^*$ found by the initialization method. Instead, we use GD on the corresponding free parameters, i.e., we use \texttt{PyTorch} to solve
\begin{align}
    \min_{W,V} &\quad \mathcal{L}_{\textit{init}}\left(A, A^*\right) \label{eq:obj A}\\
    \text{s.t.} &\quad  A\text{ as in~\eqref{eq:A generic}} \;, \label{eq:cstnt A}
\end{align}
and initialize $A$ close to $A^*$. $\mathcal{L}_\textit{init}$ is the desired loss function, also defined as the MSE throughout this work. As mentioned before, GD on such a nonconvex objective is not guaranteed to converge, and the local optimum found through this procedure might hence induce a performance drop compared to the initialization method.

 \section{Numerical Experiments}
    \label{sec:results}

To showcase the ability of SIMBa to handle different systems and both simulated and real-world data, this Section presents an extensive performance analysis, benchmarking SIMBa with existing state-of-the-art linear state-space SI methods. Our results demonstrate the efficiency of the proposed multi-step SI approach leveraging Proposition~\ref{prop:generic}, showing that it can enforce stability without sacrificing accuracy.

Throughout our analyses, we compare SIMBa to standard state-space methods implemented either in the Python-based \texttt{SIPPY} library~\cite{armenise2018open} or the MATLAB SI toolbox~\cite{ljung1995system}. Note that \texttt{SIPPY} assumes that $x_0=0$, so that results between both implementations are comparable in this setting, 
which is the case in Section~\ref{sec:res random}. When dealing with real-world input-output data in Section~\ref{sec:res io}, however, enforcing $x_0=0$ leads to suboptimal performance, and we compare SIMBa with MATLAB's performance when $x_0$ is estimated.\footnote{The code and data used for this paper can be found on \href{https://gitlab.nccr-automation.ch/loris.dinatale/simba-ecc}{https://gitlab.nccr-automation.ch/loris.dinatale/simba-ecc}.} \new{Note that we set `\texttt{Focus}' = `\texttt{simulation}' when using the MATLAB SI toolbox for a fair comparison with SIMBa's multi-step-ahead identification framework.}

    \subsection{Comparison using random stable models}
    \label{sec:res random}
    
To assess the performance of SIMBa on standard SI problems, we started by generating $50$ random stable discrete-time state-space models, from which we simulated one trajectory of $300$ steps, starting from $x_0=0$, for the training, validation, and testing data, respectively. For the three trajectories, each dimension of $u\in\mathbb{R}^m$ was generated as a Generalised Binary Noise (GBN) signal with a switching probability of $0.1$~\cite{armenise2018open}. \new{Before fitting the methods, w}\old{W}e then added white \old{output }noise $v\sim \mathcal{N}(0, 0.25)$ to the \new{output} training data. For this experiment, we arbitrarily chose $n=5$, $m=3$, $p=3$, and the default \texttt{parameters} for SIMBa, except for the number of epochs, increased to \numprint{50000} to ensure convergence. 

The performance of each SI method on the testing trajectory is plotted in Fig.~\ref{fig:random}, where green indicates SIMBa, blue other stable SI methods, and red PARSIMs, which cannot enforce stability. For each system, the MSE of each method was normalized with respect to the best-attained performance by any approach to generate the box plots, and the corresponding key metrics are reported in Table~\ref{tab:random}. For a better visual representation, we overlaid the corresponding clouds of points, where we added random noise on the x-axis to distinguish them better. Note that this zoomed-in plot does not show one instance where SIMBa did not converge and attained poor performance, while it discards three such instances for PARSIM-K and many points with a normalized MSE between $3$ and $7$ for the methods in blue.

\begin{figure}
    \begin{center}
    \includegraphics[width=\columnwidth]{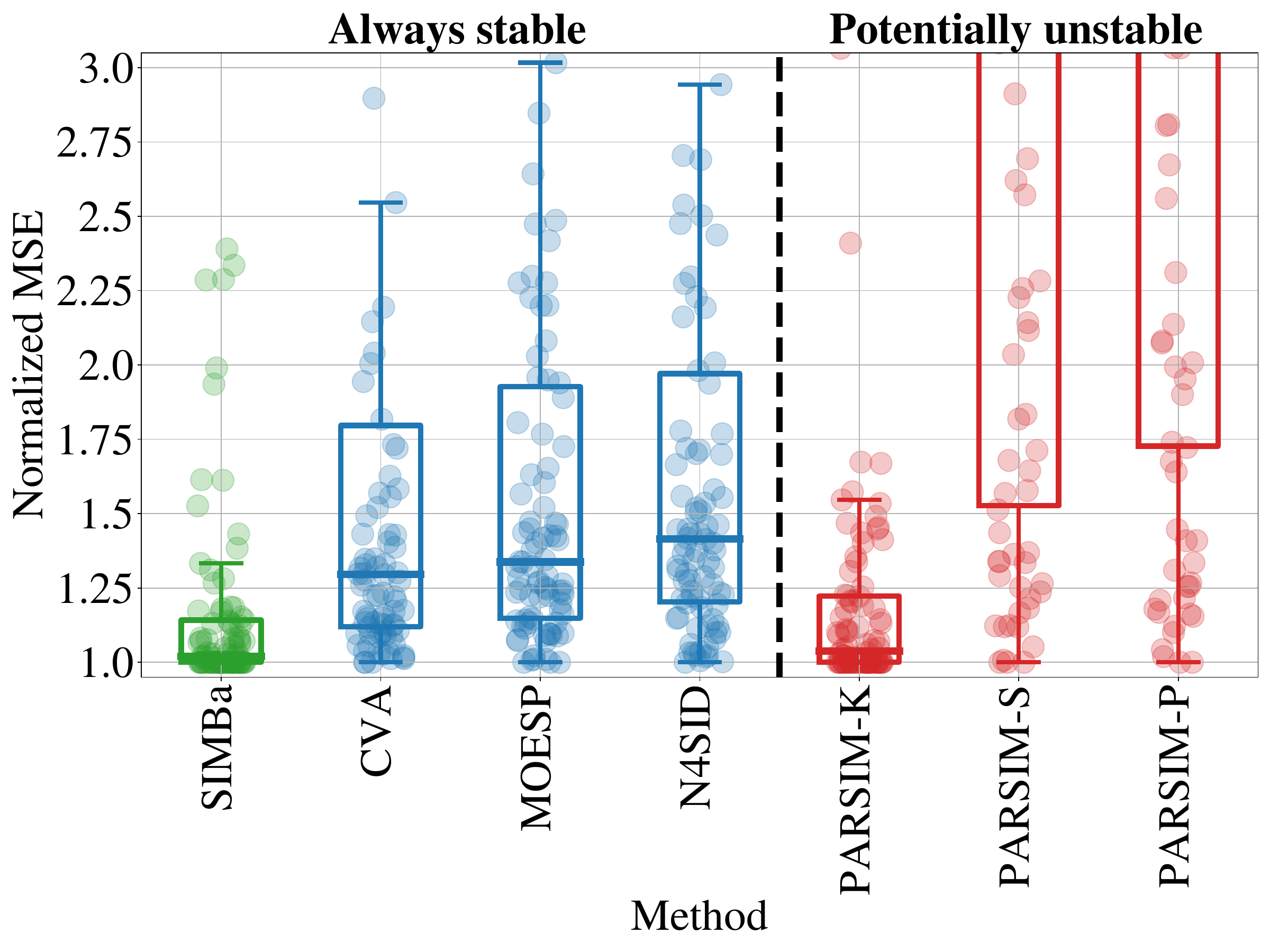}
    \caption{Performance of input-output state-space SI methods on $50$ randomly generated systems, where the MSEs have been normalized by the best-obtained error for each system. The performance of SIMBa (ours) is plotted in green, other stable SI methods in blue, while red indicates methods without stability guarantees. Key metrics are reported in Table~\ref{tab:random} for clarity.}
    \label{fig:random}
    \end{center}
\end{figure}

\begin{table}[!t]
    \caption{Normalized MSE of each method compared to the best one, reported from Fig.~\ref{fig:random}.}
    \label{tab:random}
    \centering
    \begin{tabular}{l|c|c|c}
    \hline
    \textbf{Method} & \textbf{$0.25$-quantile} & \textbf{Median} & \textbf{$0.75$-quantile} \\
    \hline
    SIMBa		&	\textbf{1.00}	&	\textbf{1.02}	&	\textbf{1.14}	\\
    CVA		&	1.12	&	1.30	&	1.80	\\
    MOESP		&	1.15	&	1.34	&	1.93	\\
    N4SID		&	1.20	&	1.42	&	1.97	\\ \hline
    PARSIM-K	&	\textbf{1.00}	&	1.04	&	1.22	\\
    PARSIM-S	&	1.53	&	3.65	&	37.5	\\
    PARSIM-P	&	1.73	&	5.37	&	84.2	\\
    \hline
    \end{tabular}
\end{table}

Overall, SIMBa shows the most robust performance, with $75\%$ of its instances achieving an error within $14\%$ of the best performance and half of them being near-optimal (see Table~\ref{tab:random}). The only method coming close is \text{PARSIM-K}, but its performance is slightly more spread out and it cannot guarantee stability. If we only look at other stable SI methods, their median accuracy is at least $30\%$ worse than the best one half of the time. In fact, their performance drop is more than $12\%$ three times out of four, compared to approximately one-fourth of the time for SIMBa, and their accuracy on the various systems is significantly more spread out. To summarize, SIMBa takes the best out of both worlds, simultaneously achieving state-of-the-art performance and stability guarantees.


    \subsection{Performance on real-world input-output data}
    \label{sec:res io}

In this Section, we leverage DAISY, a database for SI~\cite{de1997daisy}, to test our framework on real-world data. In particular, we analyze the performance of SIMBa in detail on the data collected in a \SI{120}{\mega\watt} power plant in Pont-sur-Sambre, France, where $m=5$ and $p=3$. It gathers $200$ data points with a sampling time of $\delta = 1228.8$ seconds. Here, we first standardized the input and output data so that each dimension is zero-mean and has a standard deviation of one.\footnote{Standardization generally has little impact on the performance, as analyzed in~\cite{tech}.} We use the first $100$ and $150$ samples for training and validation, respectively, and hold out the last $50$ ones for testing the final performance of the models.\footnote{\ldn{Since $x_0$ is estimated by SIMBa at training time, overlapping the training and validation data allows us to use the same initial state to validate its performance after each epoch. However, we let it run for $50$ more steps to assess its extrapolation capability and avoid overfitting the training data. For testing, we rely on MATLAB's \texttt{findstate} function to estimate $x_0$.}}


We investigate four variations of SIMBa, encoded in their names: an \textit{``i''} indicates instances with \texttt{init\_from\_matlab\_or\_ls=True}, and an \textit{``L''} that SIMBa was run for more epochs to ensure convergence. Specifically, the number of epochs with \textit{``L''} is pushed from \numprint{10000} to \numprint{20000} for \textit{SIMBa\_i} and from \numprint{25000} to \numprint{50000} otherwise. We set \texttt{dropout=0}, \texttt{learn\_x0=True} --- since it is unknown ---, and leave the other \texttt{parameters} at their default values. Since the true order of the system is unknown, one could leverage MATLAB's SI toolbox to first find the most appropriate $n$ and then run SIMBa to gain time. Here, we instead show that SIMBa dominates all the other methods from the MATLAB SI toolbox for \textit{any} choice of $n$. The PARSIM methods are not analyzed here since they all diverged for at least one value $n$ due to instability. Similarly, the stable methods from \texttt{SIPPY} achieved poor accuracy due to their assumption that $x_0=0$, which seems too restrictive for this data, and are thus omitted for clarity.

Fig.~\ref{fig:daisy} reports the MSE of the different methods on the testing data normalized by the best performance obtained by MATLAB's SI toolbox. This either corresponds to  \textit{N4SID}\footnote{Note that we set \texttt{N4Weight=\textquotesingle auto\textquotesingle} --- to automatically recover the best performance between the classical N4SID, CVA, and MOESP methods --- and \texttt{Focus=\textquotesingle simulation\textquotesingle} for a fair comparison with SIMBa, which is optimizing for the performance over the entire trajectory.} (red crosses) or the Prediction Error Method (PEM) (orange triangles). The latter aims at improving the performance of the model found by N4SID and is thus expected to perform better on the training data. Both randomly initialized versions of SIMBa, reported in black and green, were run with $10$ different seeds and the boxplot and clouds of points were generated as for Fig.~\ref{fig:random}. Since randomness has much less impact on initialized versions of SIMBa, we only report one instance of \textit{SIMBa\_i} and \textit{SIMBa\_iL} for clarity. Note that fitting a model with $n\geq 7$ on $100$ data samples is an ill-posed problem, with more parameters than data points. Interestingly, however, SIMBa still manages to outperform MATLAB in some cases in this overparametrized setting, especially when it is initialized from MATLAB's solution.

\begin{figure}
    \begin{center}
    \includegraphics[width=\columnwidth]{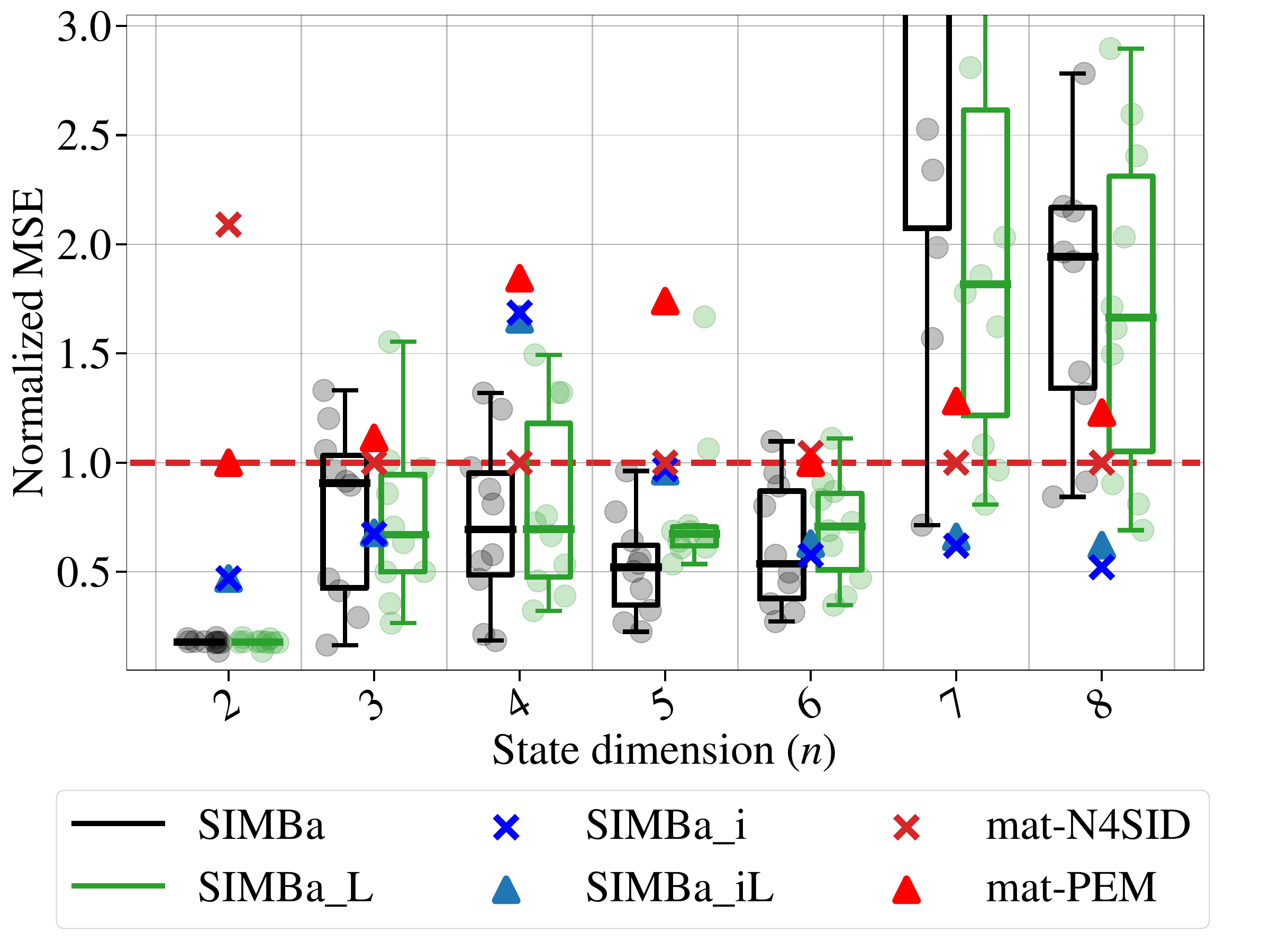}
    \caption{MSE of the different methods on the power plant test data for different choices of state dimension $n$, normalized by the best performance obtained by MATLAB's SI toolbox (red crosses and triangles). Black and green data show the performance of SIMBa over $10$ runs with random initialization for shorter (\textit{SIMBa}) and longer (\textit{SIMBa\_L}) training times, respectively. Finally, blue crosses triangles represent the performance of one initialized version of SIMBa, \textit{SIMBa\_i}, and a prolonged version \textit{SIMBa\_iL}.}
    \label{fig:daisy}
    \end{center}
\end{figure}

Impressively, SIMBa consistently attained the best performance for meaningful choices of $n\leq 6$. 
While the influence of randomness is non-negligible for randomly initialized versions, SIMBa always achieves the best accuracy, with improvements of up to $73$--$86\%$ compared to MATLAB for different choices of $n$. Furthermore, half of the time, it outperforms the latter by more than $30$--$50\%$ and $82\%$ for $n\geq 3$ and $n=2$, respectively. 

When being initialized with the solution of traditional SIMs, SIMBa always started from state-space matrices identified by MATLAB's PEM method, which achieves the best performance amongst the baselines on the validation data.\footnote{Except for $n=2$, where SIMBa was initialized from PARSIM-K.} Interestingly, SIMBa always improves PEM's performance on the testing set --- but sometimes not beyond N4SID, which attains the best accuracy amongst the baselines on this unseen data. In other words, PEM tends to overfit the training data and might start off SIMBa near a poor local minimum. While initializing with MATLAB's solution allows one to converge faster, cutting the associated computational burden, it might hence not always improve the final performance. 

    \subsection{A note on the training time}

\textit{SIMBa\_L} was run for five times more epochs than \textit{SIMBa\_i} in Section~\ref{sec:res io}, for example. Despite the small overhead required to fit $A$ during the initialization procedure in~\eqref{eq:obj A}, training \textit{SIMBa\_i} still takes approximately only $20\%$ of the time required to fit \textit{SIMBa\_L}. \ldn{In this experiment, training SIMBa ranged from $5$ to \SI{25}{\minute} \new{--- compared to a few seconds for MATLAB ---} on a MacBook Pro 2.6 GHz 6-Core Intel Core i7 laptop, irrespective of the choice of $n$. \old{Additionally, 
the training time is directly proportional to the training data length: doubling its size would double SIMBa's fitting time.} More details can be found in~\cite{tech}.}

 \section{Conclusion}
    \label{sec:conclusion}

This paper presented SIMBa, a linear state-space SI toolbox leveraging well-established ML tools to minimize the multi-step-ahead prediction error and a novel LMI-based free parametrization of Schur matrices to ensure stability. SIMBa proved to be extremely flexible, worked with various types of systems on both simulated and real-world data, and achieved impressive performance despite enforcing stability, clearly outperforming traditional stable SI methods. 

The significant average performance gains observed throughout our analyses --- ranging from approximately $25$--$30\%$ on simulated to $30$--$50\%$ on real-world data --- compared to other stable SI methods make up for the associated increased computational burden. While initializing SIMBa with the solution of traditional SI methods accelerates convergence toward meaningful solutions, it might be detrimental in the long run and get stuck in local minima. 

For future works, it would be interesting to investigate the theoretical properties of SIMBa, such as the exact data requirements needed to attain good performance or its potential integration into Koopman-based approaches. We also plan to improve the open-source toolbox towards a general tool for knowledge-grounded structured nonlinear system identification.





\bibliographystyle{IEEEtran}
\bibliography{references}

\end{document}